\documentclass{llncs}

\usepackage{epsfig}
\usepackage{amsmath}
\usepackage{color}
\usepackage{verbatim}
\usepackage{url}
\usepackage{paralist}
\usepackage{times}
\usepackage{ulem}    
\usepackage{enumitem}

\usepackage[linesnumbered,ruled,procnumbered,noend]{algorithm2e}
\usepackage{stmaryrd}
\usepackage{thmtools}
\usepackage{thm-restate}

\DeclareSymbolFont{largesymbolsA}{U}{txexa}{m}{n}
\SetSymbolFont{largesymbolsA}{bold}{U}{txexa}{bx}{n}
\DeclareFontSubstitution{U}{txexa}{m}{n}
\DeclareMathSymbol{\bigsqcupplus}{\mathop}{largesymbolsA}{"02}

\newcommand{\bigfrac}[2]{\frac{\raisebox{1ex}{$#1$}}{\raisebox{-1.5ex}{$#2$}}}

\newcommand{\forget}[1]{}

\newcommand{\redt}[1]{{\color{red}#1}}

\title{Decidability of Liveness on the TSO Memory Model}

\forget{
\author[1]{Chao Wang}
\author[2]{Gustavo Petri}
\author[3]{Yi Lv}
\author[4]{Teng Long}
\author[1,5]{Zhiming Liu}
\affil[1]{Southwest University, China}
\affil[2]{Arm Research}
\affil[3]{Institute of Software, Chinese Academy of Sciences}
\affil[4]{China University of Geosciences}
\affil[5]{Northwest Ploytechnical University, China}
}

\author {Chao Wang\inst{1} \and Gustavo Petri\inst{2} \and Yi Lv\inst{3} \and Teng Long \inst{4}
\and Zhiming Liu\inst{1,5}}
\institute{
  Southwest University, China
  \and
  Arm Research
  \and
  Institute of Software, Chinese Academy of Sciences
  \and
  China University of Geosciences
  \and
  Northwest Ploytechnical University, China
}

\begin{document}

\maketitle


\begin{abstract}
  An important property of concurrent objects is whether they support
  progress -- a special case of liveness -- guarantees, which ensure the
  termination of individual method calls under system fairness assumptions.
  Liveness properties have been proposed for concurrent objects. Typical
  liveness properties include \emph{lock-freedom}, \emph{wait-freedom},
  \emph{deadlock-freedom}, \emph{starvation-freedom} and
  \emph{obstruction-freedom}.
  It is known that the five liveness properties above are decidable on the
  Sequential Consistency (SC) memory model for a bounded number of processes.
  However, the problem of decidability of liveness for finite state concurrent
  programs running on relaxed memory models remains open.
  In this paper we address this problem for the Total Store Order (TSO) memory
  model, as found in the x86 architecture.
  We prove that lock-freedom, wait-freedom, deadlock-freedom and
  starvation-freedom are undecidable on TSO for a bounded number of processes,
  while obstruction-freedom is decidable.
\end{abstract}

\forget{
}


\forget{
}

\section{Introduction}
\label{sec:introduction}

A concurrent object provides a set of methods for client programs to access the object.
Given the complexity of writing efficient concurrent code, it is recommended to
use mature libraries of concurrent objects such as
\emph{java.util.concurrent} for Java and \emph{std::thread} for C++11.
The verification of these concurrent libraries is obviously important but intrinsically hard, since they
are highly optimized to avoid blocking -- thus exploiting more parallelism -- by
using optimistic concurrency in combination with atomic instructions like
compare-and-set.
%

Various liveness properties (progress conditons) have been proposed for concurrent objects.
Lock-freedom, wait-freedom, 
deadlock-freedom, starvation-freedom and obstruction-freedom \cite{DBLP:books/daglib/0020056,DBLP:conf/concur/LiangHFS13} are five typical liveness properties.
A liveness property describes conditions under which method calls are guaranteed to successfully complete in an execution.
Intuitively, clients of a wait-free library can expect each method call to
return in finite number of steps.
Clients using a lock-free library can expect that at any time, at least one
library method call will return after a sufficient number of steps, while it
is possible for methods on other processes to never return.
Deadlock-freedom and starvation-freedom require each fair execution to satisfy
lock-freedom and wait-freedom, respectively.
Clients using an obstruction-free library can expect each method call to
return in a finite number of steps when executed in isolation.
It is often that programmers assume that all accesses to the shared memory are
performed instantaneously and atomically, which is guaranteed only by the
sequential consistency (SC) memory model \cite{DBLP:journals/tc/Lamport79}.
However, modern multiprocessors (e.g., x86 \cite{Intel2021}, ARM \cite{ARMv8}
), and programming languages (e.g., C/C++
\cite{DBLP:conf/popl/BattyOSSW11}, Java \cite{DBLP:conf/popl/MansonPA05}) do not
implement the SC memory model. Instead they provide {\em relaxed memory models},
which allow subtle behaviors due to hardware and compiler optimizations.
For instance, in a multiprocessor system implementing the Total Store Order (TSO) memory model \cite{DBLP:conf/popl/AtigBBM10}, each processor is equipped with a FIFO store buffer.
In this paper we follow the 
TSO memory model of~\cite{DBLP:conf/popl/AtigBBM10}
(similarly
to~\cite{DBLP:conf/tphol/OwensSS09,DBLP:conf/esop/BouajjaniDM13,DBLP:conf/esop/BurckhardtGMY12}). Although in every realistic
multiprocessor system implementing the TSO memory model, the buffer is of bounded size, to describe
the semantics of \emph{any} TSO implementing system it is necessary to consider unbounded size FIFO
store buffers associated with each process, as in the semantics of~\cite{DBLP:conf/popl/AtigBBM10}. Otherwise, TSO implementations with larger store buffer will not be captured by this theoretical TSO memory model.
Any write action 
performed by a processor is put into its local store buffer
first and can then be flushed into the main memory at any time.
Some libraries 
are optimized for relaxed memory models.
For example, some work-stealing queue implementations
\cite{DBLP:conf/oopsla/LeijenSB09,DBLP:conf/ppopp/MichaelVS09} are specifically
written to perform well on TSO.

To address the problem of decidability of liveness properties, we remark that
concurrent systems with a bounded number of processes on SC can be expressed
as 
finite state $\textit{labelled transition systems}$ (LTS).
Lock-freedom, wait-freedom and obstruction-freedom can be expressed as LTL
formulas, as shown in \cite{DBLP:conf/pldi/PetrankMS09}. We show that
deadlock-freedom and starvation-freedom can 
be expressed as CTL$^*$ 
formulas.
Given that LTL and CTL$^*$ model checking is decidable 
\cite{DBLP:reference/mc/2018},
it is known that lock-freedom, wait-freedom, deadlock-freedom,
starvation-freedom and obstruction-freedom 
are decidable in this case.
However, their decidability problem on TSO memory model for a bounded number of
processes remains open. 
\forget{
{\color {red}
Liveness is an important property of programs, and {\color {blue} using
  libraries with unsuitable liveness property may cause unintended problematic
  behaviors.}
  } }


In this paper, we study the decision problem of liveness properties on TSO.
Our work covers the five typical liveness properties of
\cite{DBLP:books/daglib/0020056}, which are commonly used in practice.
Our main findings are:
\begin{itemize}
\item[-] Lock-freedom, wait-freedom, deadlock-freedom and starvation-freedom are
  \emph{undecidable on TSO}, which reveals that the verification of liveness properties
  on TSO is intrinsically harder when compared to their verification on SC.

\item[-] Obstruction-freedom is \emph{decidable on TSO}.

\end{itemize}

To the best of our knowledge, ours are the first decidability and
undecidability results for liveness properties 
on relaxed memory models.
Let us 
now present a sketch of the techniques used in the paper to justify our
findings.
\forget{
{\color{orange} GP: I think the text that follows is too detailed for an
  introduction: While the machines haven't been defined, the reader is required
  to imagine all these executions with $M_1$ and $M_2$ without even knowing what
  they represent. I would postpone this to following sections. Perhaps a couple
  of paragraphs summarizing what's here could work, but there is far too much now.}
}

\forget{
\begin{tabular}{|l|l|l|}
\hline
properties & decidability/answer on TSO & decidability/answer on SC \\
\hline
lock-freedom & $\times$ & $\checkmark$ {\color {red} TODO:add reference} \\
\hline
wait-freedom & $\times$ & $\checkmark$ {\color {red} TODO:add reference} \\
\hline
deadlock-freedom & $\times$ & $\checkmark$ {\color {red} TODO:add reference} \\
\hline
starvation-freedom & $\times$ & $\checkmark$ {\color {red} TODO:add reference} \\
\hline
obstruction-freedom & $\checkmark$ & $\checkmark$ {\color {red} TODO:add reference} \\
\hline
$k$-bounded wait-freedom & $\checkmark$ & $\checkmark$ \\
\hline
Does each wait-freedom library has \\ a bound ($n$ processes)? & $\times$ & $\checkmark$ \\
\hline
\end{tabular}
}

\vspace{5pt}
\noindent \textbf{Undecidability Proof.}
Abdulla \emph{et al.} \cite{DBLP:journals/iandc/AbdullaJ96a} reduce the cyclic
post correspondence problem (CPCP) \cite{DBLP:journals/acta/Ruohonen83}, a known
undecidable problem, into checking whether a specific lossy channel machine has
an infinite execution that visits a specific state infinitely often.
Our undecidability proof of 
lock-freedom and wait-freedom is obtained by reducing the checking of
the lossy channel machine problem into checking lock-freedom and wait-freedom for a specific library,
based on a close connection of concurrent programs on TSO and lossy
channel machines \cite{DBLP:conf/popl/AtigBBM10,DBLP:conf/atva/WangLW15}.

We 
generate a library template that can be instantiated as a specific library for each instance of CPCP. 
The collaboration between methods simulates lossy channel machine transitions.
Each execution of the lossy channel machine contains (at most) two
phases. 
Each \emph{accepting} infinite execution of the lossy channel machine loops
infinitely in 
the second phase. Thus, we make each method of the library to work differently
depending on the phase.
Our library has the following features: if an infinite library execution simulates an accepting infinite
execution of the lossy channel machine, then it violates lock-freedom, and thus,
it also violates wait-freedom; if an infinite library execution does not simulate an
accepting infinite execution of the lossy channel machine, then it satisfies
wait-freedom, and thus, it also satisfies lock-freedom. This is because any
execution that satisfies wait-freedom also satisfies lock-freedom.
Therefore, we reduce checking whether the lossy channel machine has an
accepting infinite execution, or more precisely CPCP, into checking lock-freedom
and wait-freedom of the library.

Perhaps surprisingly, the same library can be used to show the
undecidability of deadlock-freedom (resp., starvation-freedom), which requires
that each infinite fair execution satisfies lock-freedom (resp.,
wait-freedom). This is because whenever a library execution simulates an
accepting infinite execution of the lossy channel machine, we require 
library methods to collaborate and work alternatingly, and thus, such
library execution must be fair and violates deadlock-freedom, and thus,
violates starvation-freedom. Therefore, checking the existence of accepting
infinite executions of the lossy channel machines is reduced into checking
violations of deadlock-freedom and starvation-freedom.

\vspace{5pt}
\noindent \textbf{Decidability Proof.}
We introduce a notion called \emph{blocking pair}, coupling a process
control state and a memory valuation, and capturing a time point from which
we can generate an infinite execution on the SC memory model, for
which eventually one process runs in isolation and does not perform any
return. We reduce checking obstruction-freedom into 
the state reachability problem, a known decidable problem
\cite{DBLP:conf/popl/AtigBBM10}, for configurations that ``contain a blocking
pair'' and have an empty buffer for each process.
There are two difficulties here:
\begin{inparaitem}
\item[] firstly, the TSO concurrent systems of \cite{DBLP:conf/popl/AtigBBM10} do
  not use libraries; and
\item[] secondly, the state reachability problem requires the buffer of each
  process to be empty for the destination configuration, while such
  configuration may not exist in a obstruction-freedom violation.
\end{inparaitem}

The first difficulty is addressed by making each process repeatedly call
an arbitrary method with an arbitrary arguments for an arbitrarily number of
times, while transforming each call and return action into internal actions.
To solve the second difficulty, we show that each obstruction-freedom
violation has a prefix reaching a configuration that ``contains a blocking
pair''. By discarding specific actions of the prefix execution and forcing
some flush actions to happen, we obtain another prefix execution reaching a
configuration that both ``contains a blocking pair'' and has an empty buffer
for each process.

\vspace{5pt}
\noindent {\bf Related work.}
\noindent There are several works on the decidability of verification on TSO. Atig \emph{et al.} \cite{DBLP:conf/popl/AtigBBM10} prove that the state reachability problem is decidable on TSO while the repeated state reachability problem is undecidable on TSO. Bouajjani \emph{et al.} \cite{DBLP:conf/esop/BouajjaniDM13} prove that robustness is decidable on TSO.
Our previous work \cite{DBLP:conf/atva/WangLW15} proves that TSO-to-TSO linearizability \cite{DBLP:conf/esop/BurckhardtGMY12}, a correctness condition of concurrent libraries on TSO, is undecidable on TSO.
Our previous work \cite{DBLP:conf/sofsem/WangLW16} proves that a bounded version of TSO-to-SC linearizability
\cite{DBLP:conf/wdag/GotsmanMY12} is decidable on TSO.
None of these works address the decidability of concurrent library liveness on TSO.

Our approach for simulating executions of lossy channel machines with libraries is partly inspired by 
Atig \emph{et al.} \cite{DBLP:conf/popl/AtigBBM10}. 
However, Atig \emph{et al.} do not consider libraries, and their concurrent
programs do not have call or return actions. Our library needs to ensure that,
in each infinite library execution simulating an infinite execution of the lossy
channel machine, methods are ``fixed to process'', in other words, the same
method must run on the same process. The TSO concurrent systems of
\cite{DBLP:conf/popl/AtigBBM10} 
do not need to ``fix methods to processes'' since they record the control states
and transitions of each process.
Both Atig \emph{et al.} and our work store the lossy channel content in
the store buffer. However, when simulating one lossy channel machine
transition, methods of Atig \emph{et al.} only require to do one read or write
action, while methods of our paper require to read the whole channel content.
It appears that we can not ``fix methods to processes'' with methods in the
style of \cite{DBLP:conf/popl/AtigBBM10}, unless we use specific command to
directly obtain the process identifier. 

Our previous work \cite{DBLP:conf/atva/WangLW15}
considers safety properties of libraries. 
Both \cite{DBLP:conf/atva/WangLW15} and this paper 
use the collaboration of two methods to simulate one lossy channel machine transition. 
Our idea for simulating lossy channel machine transitions with libraries extends that of \cite{DBLP:conf/atva/WangLW15}, since each library constructed using the latter contains executions violating liveness, which makes such libraries not suitable for their reduction to liveness.
\forget{
The methods in \cite{DBLP:conf/atva/WangLW15}
induce infinite loops when methods do not observe some key update, and this may occur in simulating each lossy channel machine transition. Although this does not influence safety, when simulating the lossy channel machine of \cite{DBLP:journals/iandc/AbdullaJ96a}, this introduces ``false negatives'' to four liveness properties. In our paper, 
we modify the methods of \cite{DBLP:conf/atva/WangLW15} by exhausting the
buffered items in each method. This eliminates the ``false negatives''. 
}
The library of 
\cite{DBLP:conf/atva/WangLW15} contains a method that never returns and thus, do not need to consider ``fixing methods to processes''.

Our previous work \cite{DBLP:conf/sofsem/WangLW16} verifies bounded TSO-to-SC linearizability by reducing it into another known decidable reachability problem, the control state reachability problem of lossy channel machines.
That work focuses on dealing with call and return actions across multiple processes, while our verification approach for obstruction-freedom 
considers call and return action as internal actions.

\section{Concurrent Systems}
\label{sec:concurrent systems}


\subsection{Notations}

In general, a finite sequence on an alphabet $\Sigma$ is denoted $l=a_1 \cdot a_2 \cdot \ldots \cdot a_k$, where $\cdot$ is the concatenation symbol and $a_i\in\Sigma$ for each $1 \! \leq \! i \! \leq \! k$. Let $|l|$ and $l(i)$ denote the length and the $i$-th element of $l$, respectively, i.e., $|l|=k$ and $l(i)=a_i$ for $1 \! \leq \! i \! \leq \! k$.
Let $l(i,j)$ denote the string $l(i) \cdot \ldots \cdot l(j)$. 
Let $l \uparrow_{\Sigma'}$ denote the projection of $l$ on the alphabet $\Sigma'$. Given a function $f$, let $f[x:y]$ be the function that is the same as $f$ everywhere, except for $x$, where it has the value $y$. Let $\_$ denote an item, of which the value is irrelevant, and $\epsilon$ the empty word.

A $\textit{labelled transition system}$ (LTS) is a tuple $\mathcal{A}=(Q,\Sigma,\rightarrow,q_0)$, where $Q$ is a set of states, $\Sigma$ is an alphabet of transition labels, $\rightarrow\subseteq Q\times\Sigma\times Q$ is a transition relation and $q_0$ is the initial state.
A finite path of $\mathcal{A}$ is a finite  sequence of transitions
$q_0\xrightarrow{a_1}q_1\overset{a_2}{\longrightarrow}\ldots\overset{a_k}{\longrightarrow}q_k$ with
$k \! \geq \! 0$, and a finite trace of $\mathcal{A}$ is a finite sequence $t= a_1 \cdot a_2 \cdot
\ldots \cdot a_k$, with $k \! \geq \! 0$ if there exists a finite path
$q_0\overset{a_1}{\longrightarrow}q_1\overset{a_2}{\longrightarrow}\ldots\overset{a_k}{\longrightarrow}q_k$ of $\mathcal{A}$.
An infinite path of $\mathcal{A}$ is an infinite sequence of transitions
$q_0\xrightarrow{a_1}q_1\overset{a_2}{\longrightarrow}\ldots$, and correspondingly an infinite trace of $\mathcal{A}$ is an infinite sequence $t= a_1 \cdot a_2 \cdot \ldots$ if there exists an infinite path $q_0\overset{a_1}{\longrightarrow}q_1\overset{a_2}{\longrightarrow}\ldots$ of $\mathcal{A}$.

\subsection{Concurrent Objects 
and The Most General Client}

Concurrent objects 
are implemented as well-encapsulated libraries. 
The \emph{most general client} of a concurrent object 
is a program that interacts with the object, 
and is designed to exhibit all the possible behaviors of the object. 
A simple instance of the %
most general
client is a client 
that repeatedly makes non-deterministic method calls with non-deterministic arguments. 
Libraries may contain private memory locations for their own uses. For simplicity, and without loss
of generality, we assume that methods have only one argument and one return value (when 
they return).

Given a finite set $\mathcal{X}$ of memory locations, a finite set $\mathcal{M}$
of method names and a finite data domain $\mathcal{D}$, the set $\textit{PCom}$
of primitive commands
is defined by the following grammar:
\vspace{-3pt}
\[
  \begin{array}{lcl}
    \textit{PCom} & ::= & \tau \ | \ 
                          \textit{read}(x,a) \  |\  \textit{write}(x,a) \ | \ 
                        \textit{cas}\_\textit{suc}(x,a,b) \ | \
                         \textit{cas}\_\textit{fail}(x,a,b)\\
    & | & \textit{call}(m,a) \ |\ \textit{return}(m,a)
  \end{array}
\]


\noindent where $a, b \in \mathcal{D}, x \in \mathcal{X}$ 
and $m\in\mathcal{M}$. 
Here $\tau$ represents an internal command.
To use the commands as labels in an LTS we assume that they encode the
expected values that they return (an oracle of sorts). Hence, for instance the
read command $\textit{read}(x, a)$ encodes the value read $a$.
%
In general \textit{cas} (compare-and-set) commands execute a read
and a conditional write (or no write at all) in a single atomic step. In our
case a successful $\textit{cas}$ is represented with the command
$\textit{cas}\_\textit{suc}(x,a,b)$, and it is enabled when the initial value of
$x$ is $a$, upon which the command 
updates it with value $b$, while a failed
$\textit{cas}$ command, represented with the command
$\textit{cas}\_\textit{fail}(x,a,$ $b)$ does not update the state, and can only
happen when the value of $x$ is not $a$.

A library $\mathcal{L}$ is a tuple
\mbox{$\mathcal{L}$ = $(\mathcal{X}_{\mathcal{L}},\mathcal{M}_{\mathcal{L}}, \mathcal{D}_{\mathcal{L}}, Q_\mathcal{L},$ $\rightarrow_\mathcal{L})$}, where $\mathcal{X}_{\mathcal{L}}$, $\mathcal{M}_{\mathcal{L}}$ and $\mathcal{D}_{\mathcal{L}}$ are a finite memory location set, a finite method name set and a finite data domain of $\mathcal{L}$, respectively.
$Q_\mathcal{L} = \bigcup_{m \in \mathcal{M_\mathcal{L}}} Q_m $ is the union of disjoint finite sets $Q_m$ of program positions of each method $m\in\mathcal{M}_\mathcal{L}$.
Each program position represents the current program counter value and local register value of a
process and can be considered as a state.
$\rightarrow_\mathcal{L} = \bigcup_{m \in \mathcal{M}_\mathcal{L}} \rightarrow_m$ is the union of disjoint transition relations of each method $m\in\mathcal{M}_\mathcal{L}$. Let $\textit{PCom}_{\mathcal{L}}$ be the set of primitive commands (except call and return commands) upon $\mathcal{X}_{\mathcal{L}}$, $\mathcal{M}_{\mathcal{L}}$ and $\mathcal{D}_{\mathcal{L}}$. Then, for each $m \in \mathcal{M}_{\mathcal{L}}$, $\rightarrow_m \subseteq Q_m \times \textit{PCom}_{\mathcal{L}} \times Q_m$. For each $m \in \mathcal{M}_{\mathcal{L}}$ and $a\in\mathcal{D}_\mathcal{L}$, $Q$ contains an initial 
program position $\textit{is}_{(\textit{m,a})}$, which represents that library begins to execute method $m$ with argument $a$, and a final 
program position $\textit{fs}_{(\textit{m,a})}$ which represents that method $m$ has finished its execution and then a return action with return value $a$ can occur. There are neither incoming transitions to $\textit{is}_{(\textit{m,a})}$ nor outgoing transitions from $\textit{fs}_{(\textit{m,a})}$ in $\rightarrow_m$.


The most general client $\mathcal{MGC}$ is defined as a tuple $( \mathcal{M}_{\mathcal{C}}, \mathcal{D}_{\mathcal{C}}, Q_{\mathcal{C}},\rightarrow_{\textit{mgc}})$, where $\mathcal{M}_{\mathcal{C}}$ is a finite method name set, $\mathcal{D}_{\mathcal{C}}$ is a finite data domain, $Q_{\mathcal{C}}=\{\textit{in}_{\textit{clt}},\textit{in}_{\textit{lib}}\}$ 
is the state set, 
and $\rightarrow_{\textit{mgc}}=  \{ (\textit{in}_{\textit{clt}},\textit{call}(m,a),\textit{in}_{\textit{lib}}), (\textit{in}_{\textit{lib}},\textit{return}(m,b), \textit{in}_{\textit{clt}}), \vert m \in \mathcal{M}_{\mathcal{C}}, a,b \in \mathcal{D}_{\mathcal{C}} \} $ 
is a transition relation. State $\textit{in}_{\textit{clt}}$ 
represents that currently 
no method of library is running, and $\textit{in}_{\textit{lib}}$ 
represents that some method of library is running.

\subsection{TSO Operational Semantics}
\label{sec:operational semantics}

A concurrent system consists of $n$ processes, each of which runs the 
most general client $\mathcal{MGC}$ = $( \mathcal{M}, \mathcal{D}, \{\textit{in}_{\textit{clt}},\textit{in}_{\textit{lib}}\},\rightarrow_{\textit{mgc}})$, and all the most general clients interact with a same library $\mathcal{L}$ = $(\mathcal{X}_{\mathcal{L}},\mathcal{M}, \mathcal{D}, Q_\mathcal{L},\rightarrow_\mathcal{L})$.
In this paper we follow the 
TSO memory model of~\cite{DBLP:conf/popl/AtigBBM10}
(similarly
to~\cite{DBLP:conf/tphol/OwensSS09,DBLP:conf/esop/BouajjaniDM13,DBLP:conf/esop/BurckhardtGMY12}),
where each processor is equipped with a FIFO store buffer.
As explained in the introduction, in this 
TSO memory model, each process is associated with an unbounded FIFO store buffer.
Fence commands are used to ensure order between commands before fence and commands after fence. The TSO memory model of \cite{DBLP:conf/popl/AtigBBM10} does not include fence commands, since fence commands can be simulated with $\textit{cas}$ commands.

The operational semantics of a concurrent system (with library $\mathcal{L}$ and $n$ processes) on TSO is defined
as an LTS $\llbracket \mathcal{L}, n \rrbracket$ = $(\textit{Conf}, \Sigma,
\rightarrow,\textit{InitConf})$, 
with $\textit{Conf}$, $\Sigma$, $\rightarrow$
and $\textit{InitConf}$ described below. 

Configuration of $\textit{Conf}$ are tuples $(p,d,u)$, where $p: \{ 1, \ldots, n \} \rightarrow \{\textit{in}_{\textit{clt}}\} \cup (Q_{\mathcal{L}} \times \{\textit{in}_{\textit{lib}}\})$ 
represents the control state of each process, $d: \mathcal{X}_{\mathcal{L}} \rightarrow \mathcal{D}$ is the valuation of memory locations, and $u: \{ 1, \ldots, n\} \rightarrow (\{ (x,a)\ \vert\ x \in \mathcal{X}_{\mathcal{L}}, a \in \mathcal{D} \})^*$ is the content of each process's 
store buffer. The initial configuration $\textit{InitConf} \in \textit{Conf}$ is $(p_{\textit{init}}, d_{\textit{init}}, u_{\textit{init}})$. Here $p_{\textit{init}}$ maps each process id to $\textit{in}_{\textit{clt}}$, 
$d_{\textit{init}}$ is a valuation for memory locations in $\mathcal{X}_{\mathcal{L}}$, and $u_{\textit{init}}$ initializes each process with an empty buffer.


We denote with $\Sigma$ the set of actions defined by the following grammar:\\[-5pt]
\[
  \begin{array}{lcl}
    \Sigma & ::= & \tau(i) \  | \ 
    \textit{read}(i,x,a) \  |\  \textit{write}(i,x,a)\  | \
    \textit{cas}(i,x,a,b) \ | \
    \textit{flush}(i,x,a) \\
    & | & \textit{call}(i,m,a) \ | \  \textit{return}(i,m,a)
  \end{array}
\]

\noindent where $1 \! \leq \! i \! \leq \! n, m \in \mathcal{M}$, $x \in
\mathcal{X}_{\mathcal{L}}$ and $a,b \in \mathcal{D}$. The transition relation $\rightarrow$ is the least relation satisfying the
transition rules shown in \figurename~\ref{fig:transition relation te} for each
$1 \leq i \leq n$. The rules are explained below:
\begin{itemize}
\item[-] $\textit{Tau}$ rule: A $\tau$ transition only influences the control state of one process.


\item[-] $\textit{Read}$ rule: A function $\textit{lookup}(u,d,i,x)$ is used to search for the latest value of $x$ 
    in the buffer or the main memory, i.e.,
\begin{displaymath}
\textit{lookup(u,d,i,x)} = \left \{ \begin{array}{ll}
                      a & \textrm{if } u(i) 
                          \uparrow_{ \{ (x,b)\vert b \in \mathcal{D} \} }  =(x,a) \cdot l, \ \textit{for some sequence} \ l 
                          \\
                      d(x) & \textrm{otherwise }
                      \end{array} \right.
\end{displaymath}
where 
$\{ (x,b)\vert b \in \mathcal{D} \}$ is the set of items of $x$ in buffer. 
A $\textit{read}(i,x,a)$ action returns the latest value of $x$ in the buffer
if present, or returns the value in memory if the buffer contains no stores on
$x$.

\item[-] $\textit{Write}$ rule: 
    An $\textit{write}(i,x,a)$ action puts an item $(x,a)$ into
    the tail of its 
    store buffer.

\item[-] $\textit{Cas}\_\textit{Suc}$ and $\textit{Cas}\_\textit{Fail}$ rules: 
  A $\textit{cas}$ action atomically executes a read and a conditional write (or no write at all) if and only if the process's 
  store buffer is empty.

\item[-] $\textit{Flush}$ rule: 
    An $\textit{flush}$ action is carried out by the memory model to flush the item at the head of the process's 
    store buffer to memory at any time.

\item[-] $\textit{Call}$ and $\textit{Return}$ rules: After a $\textit{call}$ action, the current
  process 
  transits to $\textit{is}_{(\textit{m,a})}$. When the current process 
  transits to $\textit{fs}_{(\textit{m,a})}$ it can launch a $\textit{return}$ action and 
  move to 
  $\textit{in}_{\textit{clt}}$ of the most general client. 
\end{itemize}
\begin{figure}[tbp]
\[
\begin{array}{l c}
  \bigfrac{ p(i)=(q_i,\textit{in}_{\textit{lib}} 
  ) \quad q_i\ {\xrightarrow{\tau}}_{\mathcal{L}}\ q'_i } { ( p,d,u)\ {\xrightarrow{\tau(i)}}\ ( p[i:(q'_i,\textit{in}_{\textit{lib}} 
)],d,u )} {\textit{Tau}}
\end{array}
\]

\vspace{-2pt}






\vspace{-2pt}

\[
\begin{array}{l c}
  \bigfrac{ p(i)=(q_i,\textit{in}_{\textit{lib}} 
) \quad q_i\ {\xrightarrow{\textit{read}(x,a)}}_{\mathcal{L}}\ q'_i \quad \textit{lookup}(u,d,i,x)=a } { ( p,d,u)\ {\xrightarrow{\textit{read}(i,x,a)}}\
  ( p[i:(q'_i, \textit{in}_{\textit{lib}} 
)],d,u ] )} {\textit{Read}}
\end{array}
\]

\vspace{-2pt}

\[
\begin{array}{l c}
  \bigfrac{ p(i)=(q_i,\textit{in}_{\textit{lib}} 
) \quad q_i\
  {\xrightarrow{\textit{write}(x,a)}}_{\mathcal{L}}\ q'_i \quad u(i)=l } { ( p,d,u)\ {\xrightarrow{\textit{write}(i,x,a)}}\
  ( p[i:(q'_i, \textit{in}_{\textit{lib}} 
)],d,u[i:(x,a) \cdot l] )} {\textit{Write}}
\end{array}
\]

\vspace{-2pt}

\[
\begin{array}{l c}
  \bigfrac{ p(i)=(q_i, \textit{in}_{\textit{lib}} 
) \quad q_i\
  {\xrightarrow{\textit{cas}\_\textit{suc}(x,a,b)}}_{\mathcal{L}}\ q'_i \quad
  d(x)=a \quad u(i)=\epsilon } { ( p,d,u)\
{\xrightarrow{\textit{cas}(i,x,a,b)}}\
  ( p[i:(q'_i, \textit{in}_{\textit{lib}} 
)],d[x:b],u)} {\textit{Cas}\_\textit{Suc}}

\end{array}
\]

\vspace{-2pt}

\[
\begin{array}{l c}
  \bigfrac{ p(i)=(q_i, \textit{in}_{\textit{lib}} 
) \quad q_i\
  {\xrightarrow{\textit{cas}\_\textit{fail}(x,a,b)}}_{\mathcal{L}}\ q'_i \quad
  d(x) \neq a \quad u(i)=\epsilon } { (p,d,u)\
  {\xrightarrow{\textit{cas}(i,x,a,b)}}\ ( p[i:(q'_i, \textit{in}_{\textit{lib}} 
  )],d,u)} {\textit{Cas}\_\textit{Fail}}
\end{array}
\]

\vspace{-2pt}

\[
\begin{array}{l c}
\bigfrac{ u(i)=l \cdot (x,a) } { ( p,d,u)\
{\xrightarrow{\textit{flush}(i,x,a)}}\
( p,d[x:a],u[i:l] )} {\textit{Flush}}
\end{array}
\]

\vspace{-2pt}

\[
\begin{array}{l c}
  \bigfrac{ p(i)= \textit{in}_{\textit{clt}} 
} { (p,d,u)\
{\xrightarrow{\textit{call}(i,m,a)}}\
  ( p[ i: (\textit{is}_{(\textit{m,a})}, \textit{in}_{\textit{lib}} 
) ],d,u)} {\textit{Call}}
\end{array}
\]

\vspace{-2pt}

\[
\begin{array}{l c}
  \bigfrac{ p(i)=(\textit{fs}_{(\textit{m,a})}, \textit{in}_{\textit{lib}} 
) } { (p,d,u)\
{\xrightarrow{\textit{return}(i,m,a)}}\
  ( p[i: \textit{in}_{\textit{clt}} 
],d,u )} {\textit{Return}}
\end{array}
\]

\caption{Transition Relation $\rightarrow$}\label{fig:transition
relation te}

\vspace{-15pt}

\end{figure}

\section{Liveness}
\label{sec:liveness}

We use $T_{\omega}(\llbracket \mathcal{L}, n \rrbracket)$ to denote all 
infinite traces of the concurrent system $\llbracket \mathcal{L}, n \rrbracket$.
Given an execution $t \in T_{\omega}(\llbracket \mathcal{L}, n \rrbracket)$, 
we say a call action $t(i)$ matches a return action $t(j)$ with $i<j$, if the two actions are by the
same process, and there are no call or return actions by the same process in-between.
Here we assume that 
methods do not call other methods. 
Let $\textit{pend}\_\textit{inv}(t)$ denote the set of pending call actions of
$t$, in other words, call actions of $t$ with no matching return action in $t$.

We define the following predicates borrowed from
\cite{DBLP:conf/concur/LiangHFS13}. Since we do not consider aborts, and we do
not consider termination markers, 
we slightly modify the predicates definition of
\cite{DBLP:conf/concur/LiangHFS13} by consider only infinite executions. Given
an infinite execution $t \in T_{\omega}(\llbracket \mathcal{L},n \rrbracket)$:
\begin{itemize}
\item[-] $\textit{prog-t}(t)$: This predicate holds when every method call in $t$ eventually
  returns. Formally, for each index $i$ and action $e$, if $e \in
  \textit{pend}\_\textit{inv}(t(1,i))$, 
  there exists $j>i$, such that $t(j)$ matches $e$. 

\item[-] $\textit{prog-s}(t)$: This predicate holds when 
    there is always some method return action happens in the future if the system executes for a sufficient number of steps. Formally, for each index $i$ and action
  $e$, if $e \in \textit{pend}\_\textit{inv}(t(1,i))$ holds, then there exists $j>i$, such that
  $t(j)$ is a return action.

\item[-] $\textit{sched}(t)$: This predicate holds if 
  $t$ is an infinite trace with pending call actions, 
  and at least one of the processes with a pending call action 
  is scheduled infinitely many
  times. Formally, if $\vert t \vert = \omega$ and $\textit{pend}\_\textit{inv}(t) \neq \emptyset$,
  then there exists $e \in \textit{pend}\_\textit{inv}(t)$, such that $\vert t \uparrow_{pid(e)}
  \vert = \omega$. Here $t \uparrow_{pid(e)}$ represents the projection of $t$ into the actions of the process of $e$.

\item[-] $\textit{fair}(t)$: This predicate describing fair interleavings requires that if $t$ is an
  infinite execution, then each 
  process is scheduled infinitely many times. Formally,
  if $\vert t \vert = \omega$, then for each process $proc$ 
  in $t$,
  \mbox{$\vert t \uparrow_{proc} \vert$} \mbox{$= \omega$}. Here $t \uparrow_{proc}$ represents the
  projection of $t$ into actions of process $proc$.

\item[-] $\textit{iso}(t)$: This predicate requires that if $t$ is an infinite execution, eventually
  only one process is scheduled. Formally, if $\vert t \vert = \omega$, then there exists index $i$
  and process $proc$, such that for each $j>i$, $t(j)$ is an action of process $proc$.
\end{itemize}

With these predicates, we can present the formal notions of lock-freedom, wait-freedom, deadlock-freedom, starvation-freedom and obstruction-freedom of \cite{DBLP:conf/concur/LiangHFS13}. 

\begin{definition}\label{def:lock-free, wait-free, deadlock-free, starvation-free and obstruction-free}
Given an execution $t \in T_{\omega}(\llbracket \mathcal{L},n \rrbracket)$:
\begin{itemize}
\item[-] $t$ satisfies lock-freedom whenever: $\textit{sched}(t) \Rightarrow \textit{prog-s}(t)$,

\item[-] $t$ satisfies wait-freedom whenever: $\textit{sched}(t) \Rightarrow \textit{prog-t}(t)$,

\item[-] $t$ satisfies deadlock-freedom whenever: $\textit{fair}(t) \Rightarrow \textit{prog-s}(t)$,

\item[-] $t$ satisfies starvation-freedom whenever: $\textit{fair}(t) \Rightarrow \textit{prog-t}(t)$,

\item[-] $t$ satisfies obstruction-freedom whenever: $\textit{sched}(t) \wedge \textit{iso}(t) \Rightarrow \textit{prog-s}(t)$
\end{itemize}

For library $\mathcal{L}$, we parameterize the definitions above over $n$ processes, and we define
their satisfaction requiring that each execution of $\llbracket \mathcal{L}, n \rrbracket$ satisfies
the corresponding liveness property.
\end{definition}

Petrank \emph{et al.} \cite{DBLP:conf/pldi/PetrankMS09} demonstrate how to formalize lock-freedom, wait-freedom and obstruction-freedom as LTL formulas.
It remains to show that deadlock-freedom and starvation-freedom for $n$ processes can be formalized
as 
CTL$^*$ formulas.
Let $A$, $F$ and $G$ be the standard modalities in CTL$^*$.
Let $P_{\textit{ret}}$ be a predicate identifying return actions, 
$P_{proc}$ be a
predicate identifying actions of process $proc$, and $P_{(r,proc)}$ 
the predicate identifying process $proc$'s
return actions. We define $\textit{fair}=(GF\ P_1) \wedge \ldots \wedge (GF\ P_n)$ to describe fair
executions of $n$ processes. Then, deadlock-freedom can be defined as the 
CTL$^*$ formula $A (\textit{fair}
\rightarrow GF\ P_{\textit{ret}})$, and starvation-freedom can be defined as the 
CTL$^*$ formula $A (\textit{fair} \rightarrow GF\ P_{(r,1)} \wedge \ldots \wedge GF\ P_{(r,n)})$.

As explained in Section \ref{subsec:equivlant characterization of obstruction-freedom}, 
concurrent system with $n$ processes on SC can be expressed as finite state LTS, and LTL and CTL$^*$ model checking is decidable 
\cite{DBLP:reference/mc/2018}, we can obtain that the above five liveness properties are decidable for SC.




\section{Undecidability of Four Liveness Properties} 
\label{sec:undecidability of lock-freedom and wait-freedom}

In this section we propose our undecidability proof of lock-freedom, wait-freedom, deadlock-freedom and starvation-freedom on TSO.

\forget{
In this section we propose our undecidability proof of lock-freedom,
wait-freedom, deadlock-freedom and starvation-freedom 
on TSO.
We first introduce the notion of lossy channel machines, the cyclic post
correspondence problem (CPCP) \cite{DBLP:journals/acta/Ruohonen83}, and a known
result of Abdulla \emph{et al.} \cite{DBLP:journals/iandc/AbdullaJ96a} which
reduces checking CPCP into checking if a specific lossy channel machine has
infinite executions of a particular form.
Then, we generate a specific library for this lossy channel machine, and further
reduce checking 
the existence of such executions of the lossy channel machines into checking
lock-freedom, deadlock-freedom, wait-freedom and starvation-freedom for this
library.
}

\subsection{Perfect/Lossy Channel 
Machines}
\label{subsec:perfect/lossy channel system}

A channel machine
\cite{DBLP:journals/iandc/AbdullaJ96a,DBLP:conf/popl/AtigBBM10} is a finite
control machine equipped with channels of unbounded size. It can perform send
and receive operations on its channels. A lossy channel machine is a channel
machine where arbitrarily many items in its channels may be lost
non-deterministically at any time and without any notification.

Let $\mathcal{CH}$ be the finite set of channel names and $\Sigma_{\mathcal{CH}}$ be a finite alphabet of channel contents.
The content of a channel is a finite sequence over $\Sigma_{\mathcal{CH}}$. A
channel operation is either a send operation $c!a$ sending the value $a$ over
channel $c$, a receive operation $c?a$ receiving $a$ over $c$, or a silent operation $nop$.
We associate with each channel operation a relation over words as follows:
Given $u \in \Sigma_{\mathcal{CH}}^*$, we have $\llbracket nop \rrbracket(u,u)$, $\llbracket c!a \rrbracket(u,a \cdot u)$ and $\llbracket c?a \rrbracket(u \cdot a,u)$.
\forget{
Let $\textit{Op}(\mathcal{CH})$ be the set of channel operations over $\mathcal{CH}$.
Given $u,u' \in \mathcal{CH} \rightarrow \Sigma_{\mathcal{CH}}^*$ two functions
that record the contents of each channel, we define the channel operation
relation, relating a channel before and after, as follows:\\[-5pt]
\[
  \llbracket nop \rrbracket(u,u)\qquad
  \llbracket c!a \rrbracket(u,u[c: a \cdot u(c)]) \qquad
  \llbracket c?a \rrbracket(u'[c: u'(c) \cdot a],u')
\]
}
A channel operation over a finite channel name set $\mathcal{CH}$ is a
mapping that associates, with each channel of $\mathcal{CH}$, a channel
operation. Let $\textit{Op}(\mathcal{CH})$ be the set of channel operations
over $\mathcal{CH}$. The relation of channel operations is extended to channel
operations over $\mathcal{CH}$ as follows: given a channel operation
$\textit{op}$ over $\mathcal{CH}$ and two functions $u,u' \in \mathcal{CH}
\rightarrow \Sigma_{\mathcal{CH}}^*$, we have $\llbracket \textit{op} \rrbracket(u,u')$, if $\llbracket \textit{op}(c) \rrbracket(u(c),u'(c))$ holds for each $c \in \mathcal{CH}$.

A $\textit{channel machine}$ is formally defined as a tuple $\textit{CM} = (Q,\mathcal{CH},\Sigma_{\mathcal{CH}},\Lambda,\Delta)$, where (1) $Q$ is a finite set of states, (2) $\mathcal{CH}$ is a finite set of channel names, (3) $\Sigma_{\mathcal{CH}}$ is a finite alphabet for channel contents, (4) $\Lambda$ is a finite set of transition labels, and (5) $\Delta \subseteq Q \times (\Lambda\cup\{\epsilon\}) \times \textit{Op}(\mathcal{CH}) \times Q$ is a finite set of transitions.
When $\textit{CM}$ is considered as a perfect channel machine, its semantics is defined as an LTS $(\textit{Conf}, \Lambda,\rightarrow,\textit{initConf})$.
A configuration of $\textit{Conf}$ is a pair $(q,u)$ where $q \in Q$ and $u:\mathcal{CH} \rightarrow \Sigma_{\mathcal{CH}}^*$. $\textit{initConf}$ is the initial configuration and all its channels are empty.
The transition relation $\rightarrow$ is defined as follows: given $q,q' \in Q$ and $u,u' \in \mathcal{CH} \rightarrow \Sigma_{\mathcal{CH}}^*$,
$(q,u) \overset{\alpha}{\longrightarrow} (q',u')$,
if there exists $op$, such that $(q,\alpha,\textit{op},q') \in \Delta$ and $\llbracket \textit{op} \rrbracket (u,u')$.
When $\textit{CM}$ is considered as a lossy channel machine, its semantics is defined as another LTS $(\textit{Conf}, \Lambda,\rightarrow',\textit{initConf})$, with transition relation $\rightarrow'$ defined as follows:
$(q,u) \overset{\alpha}{\longrightarrow}' (q',u')$,
if there exists $v,v' \in \mathcal{CH} \rightarrow \Sigma_{\mathcal{CH}}^*$,
such that (1) for each $c \in \mathcal{CH}$, $v(c)$ is a sub-word of $u(c)$, (2) $(q,v) \overset{\alpha}{\longrightarrow} (q',v')$ and (3) for each $c \in \mathcal{CH}$, $u'(c)$ is a sub-word of $v'(c)$.
Here a sequence $l_1 = a_1 \cdot \ldots \cdot a_u$ is a sub-word of another sequence $l_2=b_1 \cdot \ldots \cdot b_v$, if there exists $i_1 < \ldots < i_u$, such that $a_j=b_{i_j}$ for each $1 \leq j \leq u$.

\subsection{The Lossy Channel Machine for CPCP of Abdulla et al. \cite{DBLP:journals/iandc/AbdullaJ96a}}
\label{subsec:the lossy channel machine of CPCP in Abbdulla}

Given two sequences $l$ and $l'$, let $l =_c l'$ denote that there exists
sequences $l_1$ and $l_2$, such that $l=l_1 \cdot l_2$ and $l'=l_2 \cdot l_1$.
Given two finite sequences $\alpha_1,\ldots,\alpha_m$ and $\beta_1,\ldots,\beta_m$, where each
$\alpha_i$ and $\beta_i$ is a finite sequence over a finite alphabet, a solution of $\alpha_1,\ldots,\alpha_m$ and $\beta_1,\ldots,\beta_m$ is a nonempty sequence of
indices $i_1 \cdot \ldots \cdot i_k$, such that $\alpha_{i_1} \cdot \ldots \cdot \alpha_{i_k}$ $=_c$ $\beta_{i_1} \cdot \ldots \cdot \beta_{i_k}$.
The cyclic post correspondence problem (CPCP)
\cite{DBLP:journals/acta/Ruohonen83}, known to be undecidable, requires to answer given
$\alpha_1,\ldots,\alpha_m$ and $\beta_1,\ldots,\beta_m$, whether there exists one such solution.

Given two finite sequences $A = \alpha_1,\ldots,\alpha_m$ and $B = \beta_1,\ldots,\beta_m$ of finite sequences,
Abdulla \emph{et al.} \cite{DBLP:journals/iandc/AbdullaJ96a} generate the lossy channel machine $\textit{CM}_{(A,B)}$ 
shown in \figurename \ref{fig:M_{(A,B)}, original}.
Moreover, they prove that CPCP has a solution for $A$ and $B$, if and only if
$\textit{CM}_{(A,B)}$ has an infinite execution that visits state $s_1$ infinite times.
We point the readers to \cite{DBLP:journals/iandc/AbdullaJ96a} for
  an explanation on how 
  $\textit{CM}_{(A,B)}$ solves CPCP.

\begin{figure}[tbp]
  \centering
  \includegraphics[width=0.5\textwidth]{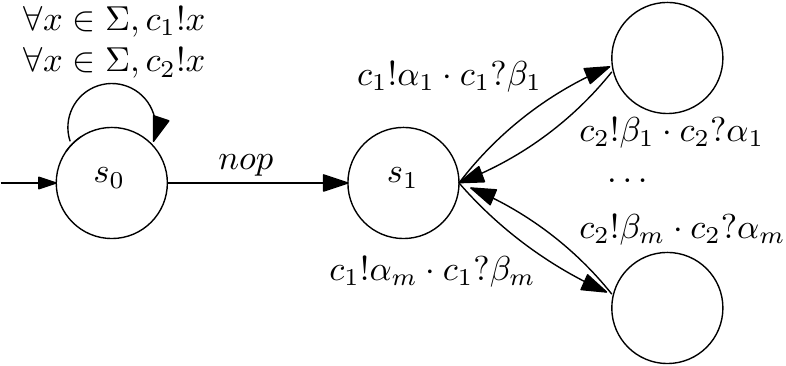}
  \caption{The lossy channel machine $\textit{CM}_{(A,B)}$.}
  \label{fig:M_{(A,B)}, original}
\end{figure}

$\textit{CM}_{(A,B)}$ contains two channels $c_1$ and $c_2$.
We use $c_1 ! \alpha_1$ to represent inserting the contents of $\alpha_1$ into $c_1$ one by one, and use $c_1 ? \beta_1$ to represent receiving the content of $\beta_1$ from $c_1$ one by one.
We use $c_1 ! \alpha_1 \cdot c_1 ? \beta_1$ to represent first do $c_1 ! \alpha_1$ and then do $c_1 ? \beta_1$.
Each execution of $\textit{CM}_{(A,B)}$ can be divided into (at most) two phases.
The first phase, called the guess phase, is a self-loop of state $s_0$, and is used to guess a solution of CPCP.
The second phase, called the check phase, goes from $s_0$ to $s_1$ and then repeatedly ``checks the content of $c_1$ and $c_2$''.

Based on $\textit{CM}_{(A,B)}$ we generate the lossy channel machine $\textit{CM}'_{(A,B)}$ which uses only one channel $c$ and works in a similar way.
To simulate one transition of $\textit{CM}_{(A,B)}$, $\textit{CM}'_{(A,B)}$ stores the content of $c_1$ followed by the content of $c_2$ (as well as new delimiter symbols) in its channel. Then it scans each symbol in its channel, modifies it (if necessary) and puts it back into its buffer, until the contents of $c_1$ and $c_2$ have all been dealt with.

\forget{
Based on $\textit{CM}_{(A,B)}$ we generate a lossy channel machine $\textit{CM}'_{(A,B)}$ which uses only one channel $c$ and works in a similar way as $\textit{CM}_{(A,B)}$. We introduce additional symbols $\bot_1$ and $\bot_2$. One transition of $\textit{CM}_{(A,B)}$ from a configuration $(q,u)$ corresponds to a sequence of transitions in $\textit{CM}'_{(A,B)}$ from configuration $(q,u')$, where $u'$ maps $c$ into $\bot_2 \cdot u(c_2) \cdot \bot_2 \cdot \bot_1 \cdot u(c_1) \cdot \bot_1$. In detail:
\begin{itemize}
\item[-] $c_1 ? a$ of $\textit{CM}_{(A,B)}$ corresponds to (1) do $c ? \bot_1 \cdot c ! \bot_1 \cdot c ? a$, then repeatedly receive non-$\bot_1$ data from channel $c$ and put it back into channel $c$, until receiving a $\bot_1$ from channel $c$ and put it back into channel $c$, and (2) do $c ? \bot_2 \cdot c ! \bot_2$, then repeatedly receive a non-$\bot_2$ data from channel $c$ and put it back into channel $c$, until receiving a $\bot_2$ and put it back into channel $c$.

\item[-] $c_1 ! a$ of $\textit{CM}_{(A,B)}$ corresponds to (1) do $c ? \bot_1 \cdot c ! \bot_1$, then repeatedly receive non-$\bot_1$ data from channel $c$ and put it back into channel $c$, until receiving a $\bot_1$ from channel $c$, and then do $c ! a \cdot c ! \bot_1$, the remaining part is the same as step (2) of $c_1 ? a$.

\item[-] The case of $c_2 ? a$ and $c_2 ! a$ is similar and we omit its detailed description here.
\end{itemize}
}

We could depict $\textit{CM}'_{(A,B)}$ similarly to \figurename~\ref{fig:M_{(A,B)},
  original}, and each transition of $\textit{CM}'_{(A,B)}$ is now a ``extended version transition''
as we discussed above. Therefore, there are ``$\textit{CM}'_{(A,B)}$'s versions'' of $s_0$, $s_1$ and
$s_{\textit{trap}}$, and when no confusion is possible we also call them $s_0$, $s_1$ and
$s_{\textit{trap}}$, respectively.
Note that if some new delimiter symbols are lost during transition, then such paths can not complete the simulation of one transition of $\textit{CM}_{(A,B)}$, and thus, do not influence the proof of the following lemma.
\forget{Note that items, including $\bot_1$ and
$\bot_2$, may be lost during transitions of $\textit{CM}'_{(A,B)}$. However,
such paths of $\textit{CM}'_{(A,B)}$ can not reach $s_1$, and thus they do not influence the proof of the following lemma.}
Based on above discussion, we reduce CPCP of $A$ and $B$ into an infinite execution problem of the lossy channel machine $\textit{CM}'_{(A,B)}$, as stated by the following lemma.

\begin{lemma}
  \label{lemma:reducing CPCP into a infinite execution problem of lossy channel machine M'{(A,B)}}
  There is a CPCP solution 
  for sequences $A$ and $B$ of finite sequences, if and only if there is an infinite execution of $\textit{CM}'_{(A,B)}$ that visits $s_1$ infinitely often.
\end{lemma}

\subsection{Libraries for Four Liveness Properties}
\label{subsec:concurrent data structures for lock-freedom and deadlock-freedom}

In this subsection, we propose our library $\mathcal{L}(A,B)$ that is generated from $\textit{CM}'_{(A,B)}$ 
and simulates the executions of $\textit{CM}'_{(A,B)}$. 
This library contains two methods $M_1$ and $M_2$. Similarly to \cite{DBLP:conf/popl/AtigBBM10,DBLP:conf/atva/WangLW15}, we use the collaboration of two methods to simulate a lossy channel. 
Our library requires that each method be fixed to a single process when simulating infinite execution of $\textit{CM}'_{(A,B)}$.
Methods of our library work differently when simulating lossy channel machine transitions of different phases. 

Let us now explain in detail the construction of $\mathcal{L}(A,B)$.
$\mathcal{L}(A,B)$ uses 
the following memory locations: 
$x_1$, $y_1$, $x_2$, $y_2$, $\textit{phase}$, $\textit{failSimu}$ and $\textit{firstM1}$.
$\textit{phase}$ stores the phase of $\textit{CM}'_{(A,B)}$, and its initial value is $\textit{guess}$.
$\textit{failSimu}$ is a flag indicating the failure of the simulation of $\textit{CM}'_{(A,B)}$, and its initial value is $\textit{false}$.
$\textit{firstM1}$ is used to indicate the first execution of $M_1$, and its initial value is $\textit{true}$.

The pseudo-code of $M_1$ and $M_2$ are shown in Algorithms \ref{Method1OfLockFreedom} and \ref{Method2OfLockFreedom}, respectively.
$\bot_s$ and $\bot_e$ are two new symbols not contained in $\textit{CM}'_{(A,B)}$.
For brevity, we use the following notations.
We use $\textit{writeOne}(x,a)$ to represent the sequence of commands writing $a$ followed by $\sharp$ into $x$.
We use $\textit{writeSeq}(x,a_1 \cdot \ldots \cdot a_k)$ to represent the sequence of commands writing $a_1 \cdot \sharp \cdot \ldots \cdot a_k \cdot \sharp$ into $x$.
$\sharp$ is a delimiter that ensures one update of a memory location will not be read twice.
\forget{For example, channel content $a \cdot b$ will be transformed into $(x,a) \cdot
(x,\sharp) \cdot (x,b) \cdot (x,\sharp)$ for some memory location $x$ in the store buffer.}
We use $v:=\textit{readOne}(x)$ to represent the sequence of commands reading $e$ followed by $\sharp$ from $x$ for some $e \neq \sharp$ and then assigning $e$ to $v$.
Moreover, if the values read do not correspond with $e$ followed by $\sharp$ we set $\textit{failSimu}$ to $\textit{true}$ and then let the current method return. This will terminate the simulation procedure.
Similarly, $v:=\textit{readRule}(x)$ reads a transition rule followed by $\sharp$ from $x$, and assign the rule to $v$.
We use $\textit{transportData}(z_1,z_2)$ to represent repeatedly using
$v=\textit{readOne}(z_1)$ to read an update of $z_1$ and using $\textit{writeOne}(z_2,v)$ to write it to $z_2$, until reading $\bot_e$ from $z_1$ and writing $\bot_e$ to $z_2$.
Given a transition rule $r$, let $valueRead(r)$ and $valueWritten(r)$ be the value received and sent by $r$, respectively. The symbols $s_0$ and $s_1$ in the pseudo-code of $M_1$ represent the corresponding state of $\textit{CM}_{(A,B)}$.

\redt{
\begin{algorithm}[t]
\KwIn {an arbitrary argument}
\While {$\textit{true}$} {
If $\textit{failSimu}$, then \KwRet; \\
\If {$\textit{firstM1}$}{
guess a transition rule $r_1$ that starts from $s_0$; \\
$\textit{writeSeq}(x_1,r_1 \cdot \bot_s \cdot \bot_e)$; \\
$\textit{firstM1}=\textit{false}$; \\
}
\Else {
$r_1:=\textit{readRule}(y_2)$; \\
let $z_1:=valueRead(r_1)$ and $z_2:=valueWritten(r_1)$; \\
$\textit{readOne}(y_2,\bot_s)$; \\
if $z_1 \neq \epsilon$, then $\textit{readOne}(y_2,z_1)$; \\
guess a transition rule $r_2$ starts from the destination state of $r_1$; \\
$\textit{writeSeq}(x_1,r_2 \cdot \bot_s)$; \\
\While{$\textit{true}$}{
$tmp:=\textit{readOne}(y_2)$; \\
if $tmp=\bot_e$, then break; \\
$\textit{writeOne}(x_1,tmp)$; \\
}
$\textit{writeSeq}(x_1,z_2 \cdot \bot_e)$; \\
}
if $\textit{phase}=\textit{guess}$ and the destination state of $r_2$ is $s_1$, then set $\textit{phase}$ to $\textit{check}$; \\
$\textit{transportData}(y_1,x_2)$; \\
if $\textit{phase}=\textit{guess}$, then \KwRet; \\
}
\caption{$M_1$}
\label{Method1OfLockFreedom}
\end{algorithm}
}

\noindent\begin{algorithm}[!h]
\KwIn {an arbitrary argument}
\While {$\textit{true}$} {
if $\textit{failSimu}$, then \KwRet; \\
$\textit{transportData}(x_1,y_1)$; \\
$\textit{transportData}(x_2,y_2)$; \\
if $\textit{phase}=\textit{guess}$, then \KwRet; \\
}
\caption{$M_2$}
\label{Method2OfLockFreedom}
\end{algorithm}

\begin{figure}[tbp]
  \centering
  \includegraphics[width=1.0\textwidth]{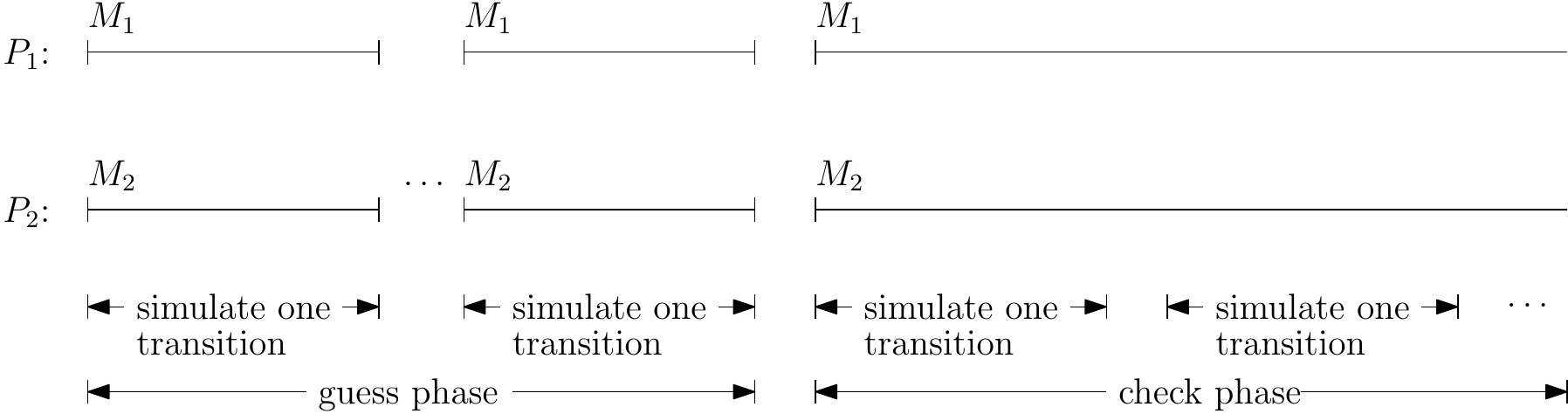}
  \caption{One execution of $\mathcal{L}(A,B)$.
  }
  \label{fig:execution for lock-freedom with one check success and one check fail}
\end{figure}

\figurename~\ref{fig:execution for lock-freedom with one check success and one
  check fail} illustrates a possible execution of $\mathcal{L}(A,B)$.
$M_1$ and $M_2$ work differently in the different phases. In the guess phase, $M_1$ and $M_2$ return after simulating one lossy channel machine transition, while in the check phase, $M_1$ and $M_2$ keep working until the simulation procedure fails.


Assume that $M_1$ (resp., $M_2$) runs on process $P_1$ (resp., process $P_2$). 
To simulate one lossy channel machine transition with channel content $l$, we first store $r_1 \cdot \bot_s \cdot l \cdot \bot_e$ in process $P_2$'s store buffer as buffered items of 
$y_2$, where $r_1$ is the transition rule of this transition, and $\bot_s$ and $\bot_e$ are additional symbols indicating the start and end of channel content of $\textit{CM}'_{(A,B)}$, respectively. Then, the procedure for simulating one lossy channel machine transition is as follows:

\begin{itemize}
\item $M_1$ reads the transition rule $r_1$ and channel content of $\textit{CM}'_{(A,B)}$ by reading all the updates of $y_2$. After reading $r_1$, $M_2$ non-deterministically chooses a transition rule $r_2$ of the lossy channel machine. Such rule should begin from the destination state of $r_1$.

\item There are four points for information update and transfer
between $M_1$ and $M_2$: (1) According to transition rule $r_1$, $M_1$ modifies and writes all the updates of $y_2$ into 
    $x_1$, (2) then $M_2$ reads all the updates of $x_1$ and writes all the updates into 
    $y_1$, (3) then $M_1$ reads all the updates of $y_1$ and writes all the updates into 
    $x_2$, and (4) finally, $M_2$ reads all the updates of $x_2$ and writes all the updates into $y_2$. To read all the updates of a memory location, we need to repeatedly read until read $\bot_e$, which indicates the end of channel content.

    Since there is no item in the buffer at the beginning of the simulation procedure, to simulate the first lossy channel machine transition $M_1$ directly writes $r \cdot \sharp \cdot \bot_s \cdot \sharp \cdot \bot_e \cdot \sharp$ to $x_1$ and does not need to read updates from $y_2$, where $r$ is a transition rule from $s_0$. This is the reason for using $\textit{firstM1}$.

\item $M_1$ is also responsible for modifying the phase (stored in the memory location $\textit{phase}$). If the last lossy channel machine transition simulated belongs to the guess phase and the destination state of $r_1$ is $s_1$, $M_1$ changes the memory location $\textit{phase}$ to check.
\end{itemize}

The reason why we need to update and transfer information between $M_1$ and $M_2$
is to deal with the case when an update of $\bot_e$ 
is not captured. 
Let us first consider a simple but infeasible solution: $M_1$ reads updates from
$y_2$ (until reading $\bot_e$), and modifies and writes all the updates into
$x_1$; while $M_2$ repeatedly reads an update of $x_1$ and writes it into $y_2$,
until reading $\bot_e$. This solution can not deal with the case when updates of
$\bot_e$ for $y_2$ is not seen by $M_1$, and will make $M_1$ and $M_2$ fall into
infinite loop that violates liveness.
This may happen
in simulating each lossy channel machine transition, and thus, introduces ``false negatives'' to four liveness properties.
To deal with this case, we need to break the infinite loop and avoid directly
writing the updates of $y_2$ back into $x_1$. Instead, in our update and
transfer points, we exhaust the updates of $y_2$, which are written to $x_1$,
and later written to $y_1$ instead of $y_2$. Therefore, there is no infinite
loop even if updates of $\bot_e$ for $y_2$ are lost.



Assume that we can successfully simulate one transition of
$\textit{CM}'_{(A,B)}$ with one $M_1$ running on process $P_1$ and one $M_2$
running on process $P_2$. Then the most general client on process $P_1$
(resp., on process $P_2$) can call $M_1$ and $M_2$. Perhaps surprisingly, the
only possible way to simulate the second transition of $\textit{CM}'_{(A,B)}$
is to let $M_1$ and $M_2$ to continue to run on processes $P_1$ and $P_2$, respectively. Let us explain why other choices fail to simulate the second transition: (1) If both processes run method $M_1$, then they both 
require reading updates of $y_1$. Since there is no buffered item for $y_1$, and none of them write to $y_1$, both $M_1$ fail the simulation.
(2) If both processes run $M_2$, we arrive at a similar situation. (3) If $M_1$
and $M_2$ run on processes $P_2$ and $P_1$, respectively. $M_1$ requires reading
the updates on $y_2$, and the only possible buffered $y_2$ items are in process $P_2$'s buffer. According to the TSO memory model, $M_1$ always reads the same value for $y_2$ and thus fails to do $\textit{readOne}(y_2,\_)$. Thus, $M_1$ fails the simulation.
Therefore, we essentially ``fix methods to processes'' without adding specific commands for checking process id.

\subsection{Undecidability of Four Liveness Properties}
\label{subsec:undecidability of four liveness properties}

The following theorem states that lock-freedom, wait-freedom,
deadlock-freedom and starvation-freedom are all undecidable on TSO for a
bounded number of processes. Perhaps surprisingly, we prove this theorem with
the same library $\mathcal{L}(A,B)$. 

\begin{theorem}
  \label{theorem:lock-freedom is undecidable}
  The problems of checking lock-freedom, wait-freedom, deadlock-freedom and starvation-freedom of a given library for a bounded number of processes 
  are undecidable on TSO.
\end{theorem}
\begin {proof}(Sketch)
  For each infinite execution $t$ of $\llbracket \mathcal{L}(A,B),2 \rrbracket$, assume that it simulates an execution of $\textit{CM}'_{(A,B)}$, or it intends to do so. There are three possible cases for $t$ shown as follows:

  \begin{itemize}
  \item[-] Case $1$: The simulation fails because some $\textit{readOne}$ does not read
    the intended value.

  \item[-] Case $2$: The simulation procedure succeeds, and $t$ infinitely loops in the guess phase.

  \item[-] Case $3$: The simulation procedure succeeds, and $t$ infinitely loops in the check phase and visits $s_1$ infinitely many times.
  \end{itemize}

  In case $1$, since $\textit{failSimu}$ is set to $\textit{true}$, each method
  returns immediately. Therefore, $t$ satisfies wait-freedom and thus, satisfies
  lock-freedom. $t$ can be either fair or unfair. In case $2$, since each method
  returns after finite number of steps in the guess phase, $t$ satisfies wait-freedom and thus, satisfies lock-freedom. 
  In case $3$, since $M_1$ and $M_2$ do not return in the check phase, $t$ violates lock-freedom and thus, violates wait-freedom. Since $M_1$ and $M_2$ coordinate when simulating each transition of $\textit{CM}'_{(A,B)}$, in case $2$ and $3$, $t$ must be fair.

  Therefore, we reduce the problem of checking whether $\textit{CM}'_{(A,B)}$ has an execution that visits $s_1$ infinitely often into the problem of checking whether $\mathcal{L}(A,B)$ has an infinite execution of case $3$ (which is fair and violates both wait-freedom and lock-freedom). By Lemma \ref{lemma:reducing CPCP into a infinite execution problem of lossy channel machine M'{(A,B)}}, we can see that the problems of checking lock-freedom, wait-freedom, deadlock-freedom and starvation-freedom of a given library for bounded number of processes are undecidable.
\end {proof}

We remark here that~\cite{DBLP:journals/corr/abs-2012-01067} considers
imposing liveness condition on store buffers, and requires buffered items to be
eventually flushed. Our undecidability results on liveness properties on TSO
still hold when imposing such liveness condition on store buffers, since in case
$3$ of the proof of Theorem \ref{theorem:lock-freedom is undecidable}, each item
put into buffer will eventually be flushed.

\section{Checking Obstruction-Freedom}
\label{sec:checking obstruction-freedom}

\subsection{The Basic TSO Concurrent Systems}
\label{subsec:the TSO concurrent system of DBLP:conf/popl/AtigBBM10}

Atig \emph{et al.} \cite{DBLP:conf/popl/AtigBBM10} considers the following concurrent systems on TSO: Each process runs a finite control state program that can do internal, 
read, write and $\textit{cas}$ actions, and different processes communicate via shared memory.
We use \emph{basic TSO concurrent systems} to denote such concurrent systems.

Formally, let $\Sigma(proc,\mathcal{D},\mathcal{X})$ be the set containing the $\tau$
(internal) action, 
the read actions, 
the write actions 
and the $\textit{cas}$ actions 
over memory locations $\mathcal{X}$ with data domain $\mathcal{D}$ and of process $proc$. A basic TSO concurrent system 
is a tuple $(P_1,\ldots,P_n)$, where each $P_i$ is a tuple $(Q_i,\Delta_i)$,
such that $Q_i$ is a finite control state set and $\Delta_i \subseteq Q_i \times
\Sigma(i,\mathcal{D},\mathcal{X}) \times Q_i$ is the transition relation. They
define an operational semantics similar to the one presented in Section
\ref{sec:concurrent systems}. Each configuration is also a tuple $(p,d,u)$,
where $p$ stores control state of each process, $d$ is a memory valuation and
$u$ stores buffer content of each process. We refer the reader to
\cite{DBLP:conf/popl/AtigBBM10} for a detailed description of the operational
semantics on TSO which is unsurprising, and hence omitted here.

Given a basic TSO concurrent system $(P_1,\ldots,P_n)$, two functions $p$ and
$p'$ that store control states of each process and two memory valuations $d$ and
$d'$, the state reachability problem requires to determine whether there is a
path from $(p,d,u_{\textit{init}})$ to $(p',d',u_{\textit{init}})$ in the
operational semantics, where $u_{\textit{init}}$ initializes each process with an empty buffer. Atig \emph{et al.} \cite{DBLP:conf/popl/AtigBBM10} prove
that the state reachability problem is decidable. 

\subsection{Verification of Obstruction-Freedom}
\label{subsec:equivlant characterization of obstruction-freedom}

The definition of obstruction-freedom requires checking infinite executions,
while the state reachability problem considers finite executions reaching
specific configurations. To bridge this gap, we propose a notion called
blocking pairs, which is defined on concurrent systems on the SC memory model
and captures potential obstruction-freedom violations.
Let $\llbracket \mathcal{L}, n \rrbracket_{sc}$ be the operational semantics of a concurrent system that runs on the SC memory model and contains $n$ processes.
The configurations of $\llbracket \mathcal{L}, n \rrbracket_{sc}$ coincide with the configurations of $\llbracket \mathcal{L}, n \rrbracket$ that preserve the buffer empty for each process.
When performing a write action $\llbracket \mathcal{L}, n \rrbracket_{sc}$ does
not put the item into the buffer, but directly updates the memory instead. $\llbracket \mathcal{L}, n \rrbracket_{sc}$ does not have flush actions, while other actions are unchanged from $\llbracket \mathcal{L}, n \rrbracket$. Since we use finite program positions, finite memory locations, a finite data domain, finite method names and a finite number of processes, and since we essentially do not use buffers, we observe that $\llbracket \mathcal{L}, n \rrbracket_{sc}$ is a finite state LTS.

Let us now propose the notion of blocking pairs. Given a state $q \in \{\textit{in}_{\textit{clt}}\} \cup (Q_{\mathcal{L}} \times \{\textit{in}_{\textit{lib}}\})$ (recall that $Q_{\mathcal{L}}$ is the set of program positions of library, $\textit{in}_{\textit{clt}}$ and $\textit{in}_{\textit{lib}}$ are the states of the most general client) and a memory valuation $d$,
$(q,d)$ is a blocking pair, if in $\llbracket \mathcal{L}, 1 \rrbracket_{sc}$
there exists a configuration $(p,d,u)$, such that the state of process 1 of $p$ is $q$ ($p(1)=q$), 
and there exists an infinite execution from $(p,d,u)$ and such
execution does not have a return action. This property can be expressed by
the 
CTL$^*$ formula $E ( (G\ \neg P_{\textit{ret}}) \wedge (G\ X\ P_{\textit{any}})
)$, where $E$ is the usual modality of CTL$^*$,
$P_{\textit{ret}}$ is a predicate that checks if the transition label is a
return action, 
and $P_{\textit{any}}$ returns
true 
{for any transition label.}
The following lemma reduces checking obstruction-freedom into the state reachability problem.

\begin{lemma}
\label{lemma:equivalent characterization of obstruction-freedom}
Given a library $\mathcal{L}$, there exists an infinite execution $t$ of $\llbracket \mathcal{L}, n \rrbracket$ that violates obstruction-freedom on TSO, if and only if there exists an finite execution $t'$ of $\llbracket \mathcal{L},n \rrbracket$ and a process $proc$, such that $t'$ leads to a
configuration $(p,d,u_{\textit{init}})$, where $(p(proc),d)$ is a blocking pair.
\end{lemma}

\forget{
The detailed proof can be found in Appendix \ref{subsec:appendix proof of
  lemma lemma:equivalent characterization of obstruction-freedom}. The
\emph{if} direction holds since we can obtain an obstruction-freedom violation
$t' \cdot t_1$, in which $t_1$ is the specific infinite execution in the
definition of blocking pairs. For the \emph{only if} direction, there exists
$t_1$, $t_2$ and process $proc$, such that $t=t_1 \cdot t_2$, and $t_2$
contains only actions of process $proc$. Let $\alpha_i$ be the last write
action of process $i \neq proc$ in $t_1$ that has not been flushed. Dropping
$\alpha_i$ in $t$ yields a legal execution, since $\alpha_i$ does not
influence the memory. Also, it is legal to clear buffer of process $proc$
before executing $t_2$. With the approach above, we can generate an new
execution that reaches a configuration, which has an empty buffer for each
process and contains a blocking pair, before the execution of $t_2$.
}

\begin {proof}(Sketch)
To prove the \emph{if} direction, consider the infinite executions which first act as $t'$, and then always run process $proc$ and disallowing other processes to do actions. The behaviors after $t'$ of these executions behave as the executions of $\llbracket \mathcal{L},1 \rrbracket_{sc}$ from the configuration  $(p',d,u_{\textit{init}})$, where $p'$ is a function that maps process $1$ to $p(proc)$. According to the definition of blocking pairs, there exists one such execution $t$ that violates obstruction-freedom.

The \emph{only if} direction is proved as follows: There exists $t_1$, $t_2$ and process $proc$, such that $t=t_1 \cdot t_2$, and $t_2$ contains only actions of process $proc$. Given process $i$, let $\alpha_i$ be the last write action of process $i$ in $t_1$ that has not been flushed. Let $t'_1$ be obtained from $t_1$ by removing the last non-flush action $\alpha$ of process $i$ after $\alpha_i$ for some process $i \neq proc$. $\alpha$ can not influence other process since it can not influence memory. Since the only possible actions of process $i$ after $\alpha$ is flush, the subsequent actions of process $i$ is not influenced. Therefore, $t'_1 \cdot t_2$ is an execution of $\llbracket \mathcal{L}, n \rrbracket$. By repeatedly apply above approach we obtain $t_3$ from $t_1$, such that $t_3 \cdot t_2$ violates obstruction-freedom, and for each process $i \neq proc$, all write of process $i$ has been flushed in $t_3$. Since process $proc$ runs as on SC in $t_2$ and read actions first try to read from buffer, we can see that $t_3 \cdot t_4 \cdot t_2$ is an execution of $\llbracket \mathcal{L}, n \rrbracket$, where $t_4$ flushes all the remaining items of process $proc$ in $t_1$. The configuration reached by execution $t_3 \cdot t_4$ has buffer empty for each process, and we can see that the control state and memory valuation is a blocking pair according to its definition. This completes the proof of this lemma.
\end {proof}

Since the model checking problem for CTL$^*$ formulas is decidable for finite state LTSs \cite{DBLP:reference/mc/2018}, we could compute the set of blocking pairs by first enumerating all
configurations of $\llbracket \mathcal{L}, 1 \rrbracket_{sc}$, and then use
model checking to check each of them. Thus, the configurations of the state reachability problem of Lemma \ref{lemma:equivalent characterization of obstruction-freedom} is computable. Since the the state reachability problem is decidable, we conclude that obstruction-freedom is decidable, as stated by the following theorem.

\forget{
With this basic TSO concurrent system, we reduce checking blocking pairs into checking if some configuration with $\textit{result}=1$ is reachable on the basic TSO concurrent system, which is known decidable \cite{DBLP:conf/popl/AtigBBM10}. The following theorem states that obstruction-freedom is decidable on TSO for $n$ processes. 
The proof can be found in Appendix \ref{subsec:appendix proof of theorem theorem:obstruction-freedom is decidable}. 
}
\begin{theorem}
  \label{theorem:obstruction-freedom is decidable}
  The problem of checking obstruction-freedom of a given library for bounded number of processes is decidable on TSO.
\end{theorem}

\section{Conclusion}
\label{sec:conclusion}

Liveness is an important property of programs, and using objects 
with
incorrect liveness assumptions can cause problematic behaviors.
In this paper, we prove that lock-freedom, wait-freedom, deadlock-freedom and
starvation-freedom are undecidable on TSO for a bounded number of processes by reducing a known undecidable problem of lossy channel machines to checking liveness properties of specific libraries.
This library simulates the lossy channel machine $\textit{CM}'_{(A,B)}$ and is designed to contain at most two kinds of executions: If 
methods collaborate in a fair way and the lossy channel machine execution being simulated visits state $s_1$ infinitely many times, then the library executions violate all four liveness properties; otherwise, the library executions satisfy all four liveness properties.
Therefore, one library is sufficient for the undecidability proof of four liveness properties.
Our undecidability proof reveals the intrinsic difference in liveness
verification between TSO and SC, resulting from the unbounded size of store buffers
in the TSO memory model.

Perhaps unexpectedly, we show that obstruction-freedom is decidable. 
Since each violation of obstruction-freedom eventually runs in isolation,
from some time point the violation, running on TSO, has the same behavior as on SC.
Therefore, checking whether a configuration contains a potential violation can be done by checking only this configuration itself, instead of checking all infinite executions from this configuration.
Checking obstruction-freedom is thus reduced to a known decidable reachability problem.
\forget{Our decidability results for liveness are obtained by reducing liveness to some decidable reachability problem. For the case of obstruction-freedom, each violation is of infinite length.
However, such executions run in isolation eventually, which causes to eventually
run as on the SC memory model.
Thus, whether a configuration contains a potential violation is computable.
The case of bounded wait-freedom is easier since we only need to check a finite prefix to find the violation.
We also prove that wait-freedom for libraries on TSO may not have a bound on steps for a bounded number of processes, while a wait-free library on SC must have bound on steps for a bounded number of processes.}

Other relaxed memory models, such as the memory models of ARM and POWER, 
are much weaker than TSO.
We conjecture that the undecidable liveness properties on TSO are still undecidable on ARM and POWER.
As future work, we would like to investigate the decidability of
obstruction-freedom 
on more relaxed memory models.
There are variants of liveness properties, such as $k$-bounded lock-freedom, bounded lock-freedom, $k$-bounded wait-freedom and bounded wait-freedom \cite{DBLP:conf/pldi/PetrankMS09}.
We would also like to investigate the decidability of bounded version of liveness properties on TSO and more relaxed memory models.





\bibliography{reference}

\begin{thebibliography}{10}

\bibitem{DBLP:journals/iandc/AbdullaJ96a}
Parosh~Aziz Abdulla and Bengt Jonsson.
\newblock Undecidable verification problems for programs with unreliable
  channels.
\newblock {\em Inf. Comput.}, 130(1):71--90, 1996.

\bibitem{DBLP:conf/popl/AtigBBM10}
Mohamed~Faouzi Atig, Ahmed Bouajjani, Sebastian Burckhardt, and Madanlal
  Musuvathi.
\newblock On the verification problem for weak memory models.
\newblock In Manuel~V. Hermenegildo and Jens Palsberg, editors, {\em
  Proceedings of the 37th {ACM} {SIGPLAN-SIGACT} Symposium on Principles of
  Programming Languages, {POPL} 2010, Madrid, Spain, January 17-23, 2010},
  pages 7--18. {ACM}, 2010.

\bibitem{DBLP:conf/popl/BattyOSSW11}
Mark Batty, Scott Owens, Susmit Sarkar, Peter Sewell, and Tjark Weber.
\newblock Mathematizing {C++} concurrency.
\newblock In Thomas Ball and Mooly Sagiv, editors, {\em Proceedings of the 38th
  {ACM} {SIGPLAN-SIGACT} Symposium on Principles of Programming Languages,
  {POPL} 2011, Austin, TX, USA, January 26-28, 2011}, pages 55--66. {ACM},
  2011.

\bibitem{DBLP:conf/esop/BouajjaniDM13}
Ahmed Bouajjani, Egor Derevenetc, and Roland Meyer.
\newblock Checking and enforcing robustness against {TSO}.
\newblock In Matthias Felleisen and Philippa Gardner, editors, {\em Programming
  Languages and Systems - 22nd European Symposium on Programming, {ESOP} 2013,
  Held as Part of the European Joint Conferences on Theory and Practice of
  Software, {ETAPS} 2013, Rome, Italy, March 16-24, 2013. Proceedings}, volume
  7792 of {\em Lecture Notes in Computer Science}, pages 533--553. Springer,
  2013.

\bibitem{DBLP:conf/esop/BurckhardtGMY12}
Sebastian Burckhardt, Alexey Gotsman, Madanlal Musuvathi, and Hongseok Yang.
\newblock Concurrent library correctness on the {TSO} memory model.
\newblock In Helmut Seidl, editor, {\em Programming Languages and Systems -
  21st European Symposium on Programming, {ESOP} 2012, Held as Part of the
  European Joint Conferences on Theory and Practice of Software, {ETAPS} 2012,
  Tallinn, Estonia, March 24 - April 1, 2012. Proceedings}, volume 7211 of {\em
  Lecture Notes in Computer Science}, pages 87--107. Springer, 2012.

\bibitem{DBLP:reference/mc/2018}
Edmund~M. Clarke, Thomas~A. Henzinger, Helmut Veith, and Roderick Bloem,
  editors.
\newblock {\em Handbook of Model Checking}.
\newblock Springer, 2018.

\bibitem{Intel2021}
Intel Corporation.
\newblock {\em Intel 64 and IA-32 Architectures Software Developer’s Manual}.
\newblock 2021.

\bibitem{DBLP:conf/wdag/GotsmanMY12}
Alexey Gotsman, Madanlal Musuvathi, and Hongseok Yang.
\newblock Show no weakness: Sequentially consistent specifications of {TSO}
  libraries.
\newblock In Marcos~K. Aguilera, editor, {\em Distributed Computing - 26th
  International Symposium, {DISC} 2012, Salvador, Brazil, October 16-18, 2012.
  Proceedings}, volume 7611 of {\em Lecture Notes in Computer Science}, pages
  31--45. Springer, 2012.

\bibitem{DBLP:books/daglib/0020056}
Maurice Herlihy and Nir Shavit.
\newblock {\em The art of multiprocessor programming}.
\newblock Morgan Kaufmann, 2008.

\bibitem{DBLP:journals/corr/abs-2012-01067}
Ori Lahav, Egor Namakonov, Jonas Oberhauser, Anton Podkopaev, and Viktor
  Vafeiadis.
\newblock Making weak memory models fair.
\newblock {\em CoRR}, abs/2012.01067, 2020.

\bibitem{DBLP:journals/tc/Lamport79}
Leslie Lamport.
\newblock How to make a multiprocessor computer that correctly executes
  multiprocess programs.
\newblock {\em {IEEE} Trans. Computers}, 28(9):690--691, 1979.

\bibitem{DBLP:conf/oopsla/LeijenSB09}
Daan Leijen, Wolfram Schulte, and Sebastian Burckhardt.
\newblock The design of a task parallel library.
\newblock In Shail Arora and Gary~T. Leavens, editors, {\em Proceedings of the
  24th Annual {ACM} {SIGPLAN} Conference on Object-Oriented Programming,
  Systems, Languages, and Applications, {OOPSLA} 2009, October 25-29, 2009,
  Orlando, Florida, {USA}}, pages 227--242. {ACM}, 2009.

\bibitem{DBLP:conf/concur/LiangHFS13}
Hongjin Liang, Jan Hoffmann, Xinyu Feng, and Zhong Shao.
\newblock Characterizing progress properties of concurrent objects via
  contextual refinements.
\newblock In Pedro~R. D'Argenio and Hern{\'{a}}n~C. Melgratti, editors, {\em
  {CONCUR} 2013 - Concurrency Theory - 24th International Conference, {CONCUR}
  2013, Buenos Aires, Argentina, August 27-30, 2013. Proceedings}, volume 8052
  of {\em Lecture Notes in Computer Science}, pages 227--241. Springer, 2013.

\bibitem{ARMv8}
ARM Limited.
\newblock {\em ARM Architecture Reference Manual ARMv8}.
\newblock 2013.

\bibitem{DBLP:conf/popl/MansonPA05}
Jeremy Manson, William Pugh, and Sarita~V. Adve.
\newblock The java memory model.
\newblock In Jens Palsberg and Mart{\'{\i}}n Abadi, editors, {\em Proceedings
  of the 32nd {ACM} {SIGPLAN-SIGACT} Symposium on Principles of Programming
  Languages, {POPL} 2005, Long Beach, California, USA, January 12-14, 2005},
  pages 378--391. {ACM}, 2005.

\bibitem{DBLP:conf/ppopp/MichaelVS09}
Maged~M. Michael, Martin~T. Vechev, and Vijay~A. Saraswat.
\newblock Idempotent work stealing.
\newblock In Daniel~A. Reed and Vivek Sarkar, editors, {\em Proceedings of the
  14th {ACM} {SIGPLAN} Symposium on Principles and Practice of Parallel
  Programming, {PPOPP} 2009, Raleigh, NC, USA, February 14-18, 2009}, pages
  45--54. {ACM}, 2009.

\bibitem{DBLP:conf/tphol/OwensSS09}
Scott Owens, Susmit Sarkar, and Peter Sewell.
\newblock A better x86 memory model: x86-tso.
\newblock In Stefan Berghofer, Tobias Nipkow, Christian Urban, and Makarius
  Wenzel, editors, {\em Theorem Proving in Higher Order Logics, 22nd
  International Conference, TPHOLs 2009, Munich, Germany, August 17-20, 2009.
  Proceedings}, volume 5674 of {\em Lecture Notes in Computer Science}, pages
  391--407. Springer, 2009.

\bibitem{DBLP:conf/pldi/PetrankMS09}
Erez Petrank, Madanlal Musuvathi, and Bjarne Steensgaard.
\newblock Progress guarantee for parallel programs via bounded lock-freedom.
\newblock In Michael Hind and Amer Diwan, editors, {\em Proceedings of the 2009
  {ACM} {SIGPLAN} Conference on Programming Language Design and Implementation,
  {PLDI} 2009, Dublin, Ireland, June 15-21, 2009}, pages 144--154. {ACM}, 2009.

\bibitem{DBLP:journals/acta/Ruohonen83}
Keijo Ruohonen.
\newblock On some variants of post's correspondence problem.
\newblock {\em Acta Inf.}, 19:357--367, 1983.

\bibitem{DBLP:conf/atva/WangLW15}
Chao Wang, Yi~Lv, and Peng Wu.
\newblock Tso-to-tso linearizability is undecidable.
\newblock In Bernd Finkbeiner, Geguang Pu, and Lijun Zhang, editors, {\em
  Automated Technology for Verification and Analysis - 13th International
  Symposium, {ATVA} 2015, Shanghai, China, October 12-15, 2015, Proceedings},
  volume 9364 of {\em Lecture Notes in Computer Science}, pages 309--325.
  Springer, 2015.

\bibitem{DBLP:conf/sofsem/WangLW16}
Chao Wang, Yi~Lv, and Peng Wu.
\newblock Bounded tso-to-sc linearizability is decidable.
\newblock In Rusins~Martins Freivalds, Gregor Engels, and Barbara Catania,
  editors, {\em {SOFSEM} 2016: Theory and Practice of Computer Science - 42nd
  International Conference on Current Trends in Theory and Practice of Computer
  Science, Harrachov, Czech Republic, January 23-28, 2016, Proceedings}, volume
  9587 of {\em Lecture Notes in Computer Science}, pages 404--417. Springer,
  2016.

\end{thebibliography}
\bibliographystyle{plain}






\end{document}